\documentclass[manuscript,screen]{acmart}

\usepackage{url,titlesec}
\usepackage{amsmath,amsthm,amssymb}
\usepackage{dsfont}

\titlespacing*{\section}{0pt}{0.75ex}{0.5ex}
\titlespacing*{\subsection}{0pt}{0.5ex}{0.25ex}

\AtBeginDocument{%
  \providecommand\BibTeX{{%
    \normalfont B\kern-0.5em{\scshape i\kern-0.25em b}\kern-0.8em\TeX}}}

\setcopyright{none}
\newtheorem{defn}{Definition}

\newtheorem{prop}{Proposition}
\newtheorem{cor}{Corollary}

\begin{document}

\title[Identification of Fair Auditors based on Non-Comparative Fairness Notion]{On the Identification of Fair Auditors to Evaluate Recommender Systems based on a Novel Non-Comparative Fairness Notion \vspace{-2ex}}

\author{Mukund Telukunta}
\affiliation{%
  \institution{Missouri University of Science and Technology}
  \city{Rolla, Missouri}}
\email{mt3qb@umsystem.edu}

\author{Venkata Sriram Siddhardh Nadendla}
\affiliation{%
  \institution{Missouri University of Science and Technology}
  \city{Rolla, Missouri}}
\email{nadendla@umsystem.edu}

\renewcommand{\shortauthors}{M. Telukunta and V. S. S. Nadendla}

\begin{abstract}
Decision-support systems are information systems that offer support to people's decisions in various applications such as judiciary, real-estate and banking sectors. Lately, these support systems have been found to be discriminatory in the context of many practical deployments. In an attempt to evaluate and mitigate these biases, algorithmic fairness literature has been nurtured using notions of comparative justice, which relies primarily on comparing two/more individuals or groups within the society that is supported by such systems. However, such a fairness notion is not very useful in the identification of fair auditors who are hired to evaluate latent biases within decision-support systems. As a solution, we introduce a paradigm shift in algorithmic fairness via proposing a new fairness notion based on the principle of non-comparative justice. Assuming that the auditor makes fairness evaluations based on some (potentially unknown) desired properties of the decision-support system, the proposed fairness notion compares the system's outcome with that of the auditor's desired outcome.
We show that the proposed fairness notion also provides guarantees in terms of comparative fairness notions by proving that any system can be deemed fair from the perspective of comparative fairness (e.g. individual fairness and statistical parity) if it is non-comparatively fair with respect to an auditor who has been deemed fair with respect to the same fairness notions. We also show that the converse holds true in the context of individual fairness. A brief discussion is also presented regarding how our fairness notion can be used to identify fair and reliable auditors, and how we can use them to quantify biases in decision-support systems.

\end{abstract}


\maketitle

\section{Introduction \label{sec: Introduction}}
In recent years, the rapid advancements in the fields of artificial intelligence (AI) and machine learning (ML) have resulted in the proliferation of algorithmic decision making in many practical applications such as decision-support systems for judges whether or not to release a prisoner on parole \cite{baryjester2015}, decisions for banks regarding granting or denying loans \cite{koren2016}, and product recommendations by e-commerce websites \cite{smith2017two}. Although these algorithms are highly efficient, they are found to be biased and unfair which substantially affect human lives. For example, A study by ProPublica in \cite{angwin2016} has shown that the recidivism scores computed by the COMPAS algorithm are found to be biased in terms of both race and gender. Similarly, racial discrimination is found in e-commerce services in online markets \cite{Harvard2016}, as well as in life insurance premiums \cite{waxmen2018}. In order to reduce the biases in such recommender systems, researchers have proposed different notions of fairness, including individual fairness, statistical parity, equality of opportunity and other similar notions \cite{MoritzOpportunities,chen2019fairness, zafar2017, kusner2017, ritov2017conditional}. \emph{Statistical parity} states that the prediction outcomes across different groups are ensured to have equal probabilities. Unlike some of the other formalizations of fairness, statistical parity is independent of the “ground truth” (i.e. correctness in the data). Even though such fairness notion can be achieved without making any assumptions on data, it does not ensure fairness to individuals or subgroups across different groups. For instance, if there are disproportional outcomes concerning non-protected attributes, statistical parity may show discrimination against qualified individuals in the group due to incorrect predictions. On the contrary, \emph{individual fairness} \cite{dwork2012fairness} states that similar individuals should be treated similarly concerning a particular task. Though the notion focuses on individuals rather than groups, it is hard to determine a metric to compute similarity between two individuals. However, Dwork pointed out in \cite{dwork2012fairness} that the establishment of metric on individuals already exists in many decision support systems such as experience and skills for a job, grades for college admission, and FICO scores for loan applications. Dwork emphasized that similarity metric can be constructed appropriately if the two individuals belong to the same group. On the other hand, human judgements are necessary if the two individuals belong two different groups.

Traditional algorithmic fairness notions have exhibited several drawbacks in the identification and mitigation of biases. Therefore, researchers have developed mechanisms where reliable audit feedback is collected from trusted human auditors to evaluate systems. For example, in \cite{kearns2017preventing}, Kearns \emph{et al.} proposed equality of opportunity and statistical parity notions when there are a large number of protected groups. Their work focuses on developing algorithms for learning classifiers that are fair with respect to large number of groups based on a formulation of two-player zero-sum game between a learner and an auditor. Learner strategy corresponds to classifiers that minimize the sum of prediction error and the Auditor is responsible for rectifying the fairness violation of the Learner. A similar work by Zhang and Neil \cite{zhang2016identifying} also designed audit algorithms for classifiers to identify discrimination in subgroups that have not been pre-defined. However, these approaches raise questions such as: what features should be sensitive? How will you define the infinite class of groups?. Another recent work by Kim \emph{et al.} in \cite{kim2018fairness} aims to compute a distance metric in the notion of individual fairness using an expert auditor. Auditor returns an unbiased estimator to compute the distance between two random individuals based on an unknown fairness metric. Similarly, Gillen \emph{et al.} in \cite{gillen2018online} assumes the existence of an auditor who is capable of identifying the fairness violations made in an online setting based on an unknown metric. Jung \emph{et al.} in \cite{jung2019eliciting} study an offline learning problem with subjective individual fairness which is benefited by human experts. Proposed algorithm obtains a feedback from human experts by asking them questions of the form: “should this pair of individuals be treated similarly or not?”. On the other hand, Raji \emph{et al.} in \cite{raji2020closing} suggests that auditing the algorithms can be achieved through internal organization development cycle which could help tackling the ethical issues raised in a company. Their framework is developed by a small team of auditors in a large company who present data and model documentation along with metrics to facilitate auditing in specific contexts.

\subsection{Contributions}
From a social science perspective, justice can be categorized into two principles: comparative and non-comparative justice \cite{levine2005comparative, feinberg1974noncomparative, montague1980comparative}. The principle of \emph{comparative justice} states that relatively similar cases are treated similarly and dissimilar cases are treated differently. The characteristic feature of comparative justice lies in deriving a comparison or contrast between the way in which some system has treated two individuals or groups. In contrast, the principle of \emph{non-comparative justice} mandates that each person be treated precisely as he/she deserves or merits regardless of the way in which anyone else treated. Non-comparative justice arguments do not need any comparison or contrast between two or more individuals. Although many algorithmic and data-driven approaches have been proposed to improve different fairness notions as discussed, these approaches were executed based on the \emph{principle of comparative justice}. Both individual fairness and statistical parity, either compare two individuals or two groups of individuals. This, along with strong criticism from the social scientists \cite{levine2005comparative, montague1980comparative}. To date, algorithmic fairness has only focused on the principle of comparative justice resulting in an impasse in the fairness literature, thereby raising a question as to what does fairness means from the perspective of \emph{non-comparative justice}. Therefore, in this paper, we explore a novel mathematical fairness notion using an auditor based on the principle of non-comparative justice. Auditor aims to classify the individuals using an intrinsic fair relation based on non-comparative justice. The major downside in such framework is determining whether the auditor is discriminatory or not. Hence, to compute the bias of an auditor, we investigate the relationship between comparative fairness notions (individual fairness and statistical parity) and our proposed notion of non-comparative fairness. We prove that, a system/entity satisfies comparative justice if it satisfies non-comparative fairness with respect to the auditor. We also show that converse holds true in the case of individual fairness. Once an auditor is deemed fair, his/her fair relation can be utilized to examine an unknown system.

\section{Comparative Justice for Algorithmic Fairness \label{sec: Comparative Fairness}}
Comparative justice arguments require to articulate either a comparison or a contrast between the way in which some system has treated two or more individuals or groups. The principle of comparative justice can be formulated as a combination of two standards: (1) Similar cases must be treated similarly, and (2) dissimilar cases must be treated differently. In some contexts, the system might treat individuals differently in a situation in which they should be treated similarly. This leads to comparative injustice. For instance, imagine a situation where two individuals approach a corporate bank to apply for a loan. Even though both of them have relevantly similar qualifications such as credit history, employment, and salary, the bank decides to grant the loan for only one of the individuals. This argument violates the principle of comparative justice because it violates the normative rule, "similar cases must be treated similarly". The bank should either grant both or refuse both. On the contrary, suppose there exist two industries $A$ and $B$ which cause pollution to the environment in varying amounts. Though $B$ produces twice the amount of pollution compared to $A$, the government imposes the same amount of tax on both the industries. This argument violates the principle of comparative justice as it violates the \emph{second standard}, "dissimilar cases must be treated differently".

As mentioned earlier, algorithmic fairness literature has focused on the principle of comparative justice, mainly, statistical parity and individual fairness. The notion of \emph{statistical parity} says that, the prediction outcomes across different groups are ensured to have equal probabilities. More formally,

\begin{defn}[Statistical Parity]
Given protected attributes $A$, if f is a predictor and $Y = f(X)$, where, X is the multi-attribute variable and Y is the outcome, f achieves statistical parity when
\begin{equation}
P[f(x) = y \ | \ A = a] = P[f(x) = y \ | \ A = a']
\end{equation}
holds true for all $y \in Y$, $x \in X$, and $a, a' \in A$.
\label{Defn: GF}
\end{defn}

Here, we can observe that the notion of statistical parity compares two different groups of protected attributes $a$ and $a'$. Different versions of group-conditional metric led to different notions of statistical parity. M. Hardt \emph{et al.} in \cite{MoritzOpportunities} introduced the notion \emph{equality of opportunity} which states that the true positive rate should be the same for all the groups. Similarly, Zafar \emph{et al.} in \cite{zafar2017fairness} proposes \emph{disparate mistreatment} which states that the accuracy of decisions is equal across groups of sensitive attributes with respect to false positive rates.
On the other hand, Zafar in \cite{zafar2017} says that the existing notions of fairness, based on equality in treatment, tends to limit the decision making accuracy, and recommends \emph{preference-based fairness} giving the liberty to a group of users to prefer its treatments or outcomes. Though the notion gives the liberty for a group of users, it does not emphasize on individual preferences. It is possible that few individuals in a group may prefer another option than the one preferred by most of them in the group.

On the contrary, the notion of \emph{individual fairness} \cite{dwork2012fairness} states that similar individuals should be treated similarly concerning a particular task. The notion can be defined as follows.
\begin{defn}[Individual Fairness]
Given a metric $\mathcal{D}$ for a classification task T with an outcome set $O$, a randomized predictor $f(X) = Y$ and distance measure $d : \Delta(O) \times \Delta(O) \rightarrow \mathbb{R}$, f is individually fair if and only if
\begin{equation}
d(f(x_i), f(x_j)) \leq \mathcal{D}(x_i, x_j)
\end{equation}
for all $x_i, x_j \in X$.
\label{Defn: IF}
\end{defn}

Although individual fairness compares two individuals, it does not adhere to the principle of comparative justice. In other words, the notion does not enforce on treating dissimilar individuals differently. Also, the formulation relies on a suitable similarity metric which is difficult to construct in reality. As pointed out by Speicher \emph{et al.} in \cite{speicher2018unified}, individual fairness does not regard individual's merits or desires for different decisions. Speicher proposed a way to measure individual unfairness using inequality indices and taking a person's merit into consideration. On the other hand, M. Kusner \emph{et al.} in \cite{kusner2017} developed a novel framework from causal inference, where a decision or an outcome is considered fair for an individual, only if it is the same in the actual as well as the counterfactual world. For more details, interested readers may refer to a detailed survey on fairness in algorithmic decision-making in \cite{Lepri2018, chouldechova2018frontiers}.

\section{Non-comparative Fairness in Human Judgements \label{sec: NC Fairness}}
Comparative fairness notions are defined based on a comparison between two persons, or two groups of people. However, non-comparative justice notions deviate from this assumption and define justice based on the appropriate treatment to each individual.

The principle of non-comparative justice can be formulated as, "Treat each person as he/she deserves or merits". Such an argument does not depend upon any comparison or contrast with the way in which some system treats two or more individuals or groups. For instance, historically, African-Americans are said to commit more crimes when compared to other races \cite{colorofcrime}. As a result, the COMPAS tool is biased towards African-Americans even though the severity of their crimes is less compared to other races. Since COMPAS \emph{compares} the historical information to present data, it is \emph{non-comparatively unfair}. Hence, we demonstrate the notion of non-comparative fairness with the help of an expert auditor who classifies inputs based on a fair relation regardless of the historical information or any comparisons. We define non-comparative fairness as follows.

\begin{defn}[$\epsilon$-Noncomparative Fairness]
Let $f$ denote a fair assessment (i.e. input-output relationship) of a given system, i.e. $y = f(x)$, which is evaluated subjectively by an expert auditor. If $g$ is an alternative representation, i.e. $\tilde{y} = g(x)$ (e.g. classification algorithms minimizing loss function, recommender systems), then $g$ is called $\epsilon-$noncomparatively fair w.r.t. $f$ if \begin{equation}
\displaystyle d \left( g(x), f(x) \right) < \epsilon, \text{ for all } x \in \mathcal{X}.
\end{equation}
\end{defn}

Though the above notion compares the auditor's fair assessment $f$ with the classifier/recommender $g$, it does not compare two different individuals thereby adhering to the principle of non-comparative justice. However, note that non-comparative fairness notions have their drawbacks. For example, if the auditor is discriminatory, then $f$ is no longer a fair relation. Therefore, it is necessary to investigate how traditional fairness notions are related to the notion of non-comparative fairness. One important assumption in our analysis is that we assume that both the human auditor and the algorithm employ the same distance metric $d$ in evaluating the gap between input-output relationships.

\subsection{Relation with Individual Fairness}
For the sake of this discussion, we define a weaker definition for individual fairness in the following:
\begin{defn}[Weak Individual Fairness]
Given any two individuals $x_i, x_j \in \mathcal{X}$ with $\mathcal{D} \left( x_i,  x_j \right) \geq \kappa$, then $g$ is $(\kappa, \delta)$-individually fair if $d \Big( g(x_i), g(x_j) \Big) \leq \kappa + \delta$.
\label{Defn: kappa-IF}
\end{defn}
Note that the above definition reduces to the original individual fairness notion in Definition \ref{Defn: IF} when $\delta = 0$. Recall that individual fairness partially adopts the principle of comparative fairness by comparing two individuals. In the following proposition, we show how the relation $g$ can be evaluated based on the notion of $(\kappa, \delta)$-individual fairness, when $g$ is non-comparatively fair with respect to another individually fair relation $f$.
\begin{prop}
$g$ is $(\kappa, 2 \epsilon + \delta)$-individually fair, if $g$ is $\epsilon$-noncomparatively fair with respect to $f$, and $f$ is $(\kappa, \delta)$-individually fair.
\label{Prop: (e, k, d)-IF}
\end{prop}
\begin{proof}
Given $(x_1, y_1)$ and $(x_2, y_2)$, if $\mathcal{D} \left( x_1,  x_2 \right) \geq \kappa$ (the two individuals are $\kappa$-similar), then $f$ is $(\kappa, \delta)$-individually fair if $d \Big( f(x_1), f(x_2) \Big) < \kappa + \delta$. However, note that if $g$ is $\epsilon$-noncomparatively fair with respect to $f$, then $ d \Big( g(x_1), f(x_1) \Big) < \epsilon$ and $d \Big( g(x_2), f(x_2) \Big) < \epsilon$. Therefore, by applying a chain of triangle inequalities, we obtain
\begin{equation}
\begin{array}{lcl}
d \Big( g(x_1), g(x_2) \Big) & \leq & d \Big( g(x_1), f(x_1) \Big) + d \Big(  f(x_1), f(x_2) \Big) + d \Big( f(x_2), g(x_2) \Big)
\\[2ex]
& < & 2 \epsilon + \kappa + \delta.
\end{array}
\end{equation}
\end{proof}
We illustrate this result using the following example from the banking domain. For instance, consider two individuals who are looking to apply for a loan. An individually fair banking system evaluates both the applications via collecting various customer's attributes such as gender, race, address, credit history, collateral, and his/her ability to pay back. At the same time, consider an auditor who makes fairness judgements based on the rule: "If he/she has cleared all the debts and possesses reasonably valued collateral, the loan must be granted". Given that the auditor treats any two similar individuals similarly, the auditor is individually fair. If the banking evaluation system is relatively similar to the auditor's fair relation, from Proposition \ref{Prop: (e, k, d)-IF}, the banking system is also individually fairness. Otherwise, if the banking system is not non-comparatively fair with respect to the auditor, then the system itself is also not individually fair.

\begin{prop}
If $f$ is not individually fair and if $g$ is $\epsilon$-noncomparatively fair with respect to $f$, then $g$ is also not individually fair.
\label{Prop: IF - converse}
\end{prop}
\begin{proof}
If $f$ is not individually fair, then for some input pair $(x_1, x_2)$, we have $d\left(f(x_1), f(x_2)\right) > \kappa + \delta$ for all $\kappa, \delta \in \mathbb{R}$. However, note that if $g$ is $\epsilon$-noncomparatively fair with respect to $f$, then $d\left(g(x_1), f(x_1)\right) < \epsilon$ and $d\left(g(x_2), f(x_2)\right) < \epsilon$. Therefore, by applying a chain of triangle inequalities, we have
\begin{equation}
\begin{array}{lcl}
d\left(f(x_1), f(x_2)\right) & \leq & d\left(g(x_1), f(x_1)\right) + d\left(g(x_2), f(x_2)\right) + d\left(g(x_1), g(x_2)\right)
\end{array}
\end{equation}
Substituting the bounds of $d\left(g(x_2), f(x_2)\right)$ and  $d\left(g(x_1), g(x_2)\right)$ we get
\begin{equation}
\begin{array}{lcl}
2 \epsilon + d\left(g(x_1), g(x_2)\right) & > & d\left(g(x_1), f(x_1)\right) + d\left(g(x_2), f(x_2)\right) + d\left(g(x_1), g(x_2)\right)
\\[2ex]
& \geq & d\left(f(x_1), f(x_2)\right) > \kappa + \delta
\end{array}
\end{equation}
for all $\kappa, \delta \in \mathbb{R}$. Therefore, we also have
\begin{equation}
\begin{array}{lcl}
d\left(g(x_1), g(x_2)\right) & > & \tilde{\kappa}
\end{array}
\end{equation}
for all $\tilde{\kappa} \in \mathbb{R}$.
\end{proof}

Consider the earlier example of banking where, there are two individuals, $A$ and $B$, who possess the same degree of merit. Imagine that the bank approves $A$'s loan application and denies $B$. This outcome remains the same as per the auditor's fair relation. Imagine further that neither $A$ nor $B$ merits the outcome. Though both banking's evaluation and auditor's rule seem to be similar, they violate the precept, "treat similar individuals similarly". Moreover, the outcome violates the principle of non-comparative justice, since $A$ is treated in a way that $A$ does not merit. Hence, we can assert that banking evaluation does not satisfy individual fairness.

\begin{prop}
Even though $f$ is $(\kappa, \delta)$-individually fair, $g$ is also not individually fair if $g$ is not noncomparitively fair with respect to $f$.
\label{Prop: IF - converse 2}
\end{prop}
\begin{proof}
If $g$ is not noncomparitively fair with respect to $f$, then there exists some $x \in \mathcal{X}$ such that $d\left(g(x), f(x)\right)$ is unbounded. Let $x_1$ be the input for which $d\left(g(x_1), f(x_1)\right)$ is unbounded, i.e. $d\left(g(x_1), f(x_1)\right) > \pi$ for any $\pi \in \mathbb{R}$. At the same time, let $x_2$ denote another input for which $d\left(g(x_2), f(x_2)\right)$ is bounded by $\epsilon$.

By applying a chain of triangle inequalities, we know
\begin{equation}
\begin{array}{lcl}
d\left(g(x_1), f(x_1)\right) & \leq & d\left(g(x_1), g(x_2)\right) + d\left(g(x_2), f(x_2)\right) + d\left(f(x_2), f(x_1)\right)
\end{array}
\end{equation}

Substituting the bounds for the terms $d\left(g(x_2), f(x_2)\right)$ and $d\left(f(x_1), f(x_2)\right)$, we obtain
\begin{equation}
\begin{array}{lcl}
d\left(g(x_1), g(x_2)\right) + \epsilon + \kappa + \delta & > & d\left(g(x_1), g(x_2)\right) + d\left(g(x_2), f(x_2)\right) + d\left(f(x_2), f(x_1)\right)
\\[2ex]
& > & \pi,
\end{array}
\end{equation}
for any $\pi \in \mathbb{R}$. In other words, we have
\begin{equation}
d\left(g(x_1), g(x_2)\right) > \pi - (\epsilon + \kappa + \delta),
\end{equation}
for any $\pi \in \mathbb{R}$.
\end{proof}

To illustrate Proposition \ref{Prop: IF - converse 2}, consider the recent accusation of inequality faced by the Apple Card (credit card by Apple Inc.) \cite{applecard}. The government has launched an investigation into the company's practices after a user accused Apple of gender bias in determining credit limits. The user received a limit 20 times higher than his wife's even though both the profiles are relatively similar. As per Apple's practices, similar individuals are \emph{not} treated similarly. Now, consider an auditor who determines the credit limits of the users using the rule: "If the credit score of the user is greater than $x$ and payments are paid duly, the credit line must be \$3000. If not, the credit line can be \$1000". Note that, the auditor's fair relation treats similar individuals similarly. We can observe that Apple's evaluation of determining credit limits is significantly different from auditor's rule. Therefore, Apple does not satisfy non-comparative fairness and thereby, dissatisfying individual fairness.


\subsection{Relation with Statistical Parity}
We define a weaker definition for group fairness in the following:
\begin{defn}[Weak Statistical Parity]
Given set of protected attributes $A$, $f$ satisfies $\delta$-statistical parity when
\begin{equation}
P[f(x) = y \ | \ A = a] - P[f(x) = y \ | \ A = a'] \leq \delta
\end{equation}
holds true for all $y \in Y$, $x \in X$, and $a, a' \in A$.
\label{Defn: delta-GF}
\end{defn}
Note that the above definition equates Definition \ref{Defn: GF} when the difference between the probabilities is equal to zero. As discussed earlier, statistical parity resembles the principle of comparative justice by comparing two different protected groups. In the remaining section, we will focus on the relationship between statistical parity and non-comparative fairness. For the sake of convenience, let us denote $p_{x,y}(g,a) = P[g(x) = y \ | \ A = a]$.
\begin{prop}
Given that the probability distributions are $M$-Lipschitz continuous over all possible $f$ and $g$ functions, $g$ satisfies $(2M\epsilon + \delta)$-statistical parity, if $g$ is $\epsilon$-noncomparatively fair with respect to $f$, and $f$ satisfies $\delta$-statistical parity.
\label{Prop: GF}
\end{prop}
\begin{proof}
Given the set of protected attributes $\mathcal{A}$, since $f$ satisfies $\delta$-statistical parity, we have $ || p_{x,y}(f,a) - p_{x,y}(f,a') || < \delta$ for all $a, a' \in \mathcal{A}$. Then, we have
\begin{equation}
\begin{array}{lcl}
p_{x,y}(g, a) - p_{x,y}(g, a') & = & \left[ p_{x,y}(g, a) - p_{x,y}(f, a) \right] - \left[ p_{x,y}(g, a') - p_{x,y}(f, a') \right]
\\[2ex]
&& \qquad + \left[ p_{x,y}(f, a) - p_{x,y}(f, a') \right]
\end{array}
\end{equation}
Assuming $M$-Lipschitz continuity over all $f(x)$, $g(x)$, we have $|| p_{x,y}(g, a) - p_{x,y}(f, a) || < M \cdot \epsilon$, since $d (g(x), f(x) ) < \epsilon$. Combining all the inequalities, we have $|| p_{x,y}(g, a) - p_{x,y}(g, a') || < 2 M \epsilon + \delta$.
\end{proof}

An illustrative example is the \emph{reservation system} enforced by the Indian Constitution which sets quotas to socially and educationally backward communities to offer fair representation in the field of education, employment, and politics \cite{reservationindia} from a group fairness perspective. Reservations are employed to 3 minority groups: Scheduled Castes (7.5\% quota), Scheduled Tribes (15\% quota), and Other Backward Classes (27\% quota) ranked in decreasing order of their income cap. The remaining quota of jobs/admissions are allotted to the merit group of citizens. Now, suppose an auditor presents fair judgements based on the rule: "If group A's annual income is less than \$8000, then 30\% of the quota of jobs/college admissions are allotted to A. Similarly, if group B's annual income is greater than \$15000, then 17\% of the quota is assigned to B". Note that, the auditor's fair relation is somewhat similar to that of the government's policy. Assuming the government's policy is fair, the auditor is also unbiased from a group fairness perspective.

\section{Identification of Fair Auditors \label{sec: Unknown Biases}}

Several attempts have been made to identify how people perceive fairness in order to automate the process of fixing fairness from an algorithmic standpoint \cite{binns2019human,grgic2018human,greenberg1986determinants, saxena2019fairness, srivastava2019mathematical}. However, the underlying assumption in most of the fairness literature is that these systems are evaluated by fair and unbiased auditors. This is not always true because people exhibit a wide range of biases based on diverse prior experiences. As a solution to this problem, we will demonstrate how the proposed non-comparative fairness notion can be used to identify fair auditors. A direct approach is to compare an auditor's evaluation to a benchmark entity (e.g. a well-studied system, a known human expert) whose biases are well-quantified in terms of both weak individual fairness and weak group fairness notions as defined in Section \ref{sec: NC Fairness}. Using Propositions \ref{Prop: (e, k, d)-IF} and \ref{Prop: GF}, we can evaluate the biases of an unknown auditor as a function of biases present in the benchmark entity, as stated below.
\begin{cor}[to Proposition \ref{Prop: (e, k, d)-IF}]
If the benchmark entity is $(\kappa, \delta)$-individually fair, then the unknown auditor can be deemed $(\kappa, \delta')$-individually fair when the auditor is $\epsilon$-noncomparitively fair with respect to the benchmark entity such that
$$
\epsilon < \displaystyle \frac{\delta' - \delta - \kappa}{2}.
$$
\end{cor}
\begin{proof}
Let $g$ denote the unknown auditor and $f$ denote the benchmark entity. From Proposition \ref{Prop: (e, k, d)-IF}, if $f$ is $(\kappa, \delta)$-individually fair and $g$ is $\epsilon$-noncomparatively fair with respect to $f$, we have
\begin{equation}
\begin{array}{lcl}
d\left(g(x_1), g(x_2)\right) & < & 2\epsilon + \kappa + \delta
\end{array}
\end{equation}

However, our goal is to identify auditors who are $(\kappa, \delta')$-individually fair. In other words, we need $\delta'$ to be at least as large as $2\epsilon + \kappa + \delta$. In other words,
\begin{equation}
\begin{array}{lcl}
\delta' & > & 2\epsilon + \kappa + \delta.
\end{array}
\end{equation}

Upon rearranging the terms, we get the bound on $\epsilon$ stated in this corollary.

\end{proof}

Similarly, we can also quantify the tolerable bias in an unknown auditor if the goal is to identify a fair auditor from a weak statistical parity sense. This is discussed in the following corollary.
\begin{cor}[to Proposition \ref{Prop: GF}]
Assuming $M$-Lipschitz continuity in probability distributions at both the benchmark entity and the unknown auditor, if the benchmark entity satisfies $\delta$-statistical parity, then the unknown auditor can be deemed to satisfy $\delta'$-statistical parity when the auditor is $\epsilon$-noncomparitively fair with respect to the benchmark entity such that
$$
\epsilon < \displaystyle \frac{\delta' - \delta}{2M}.
$$
\end{cor}
\begin{proof}
Let $g$ denote the unknown auditor and $f$ denote the benchmark entity. From Proposition \ref{Prop: GF}, if all probability distributions are $M$-Lipschitz continuous, $f$ satisfies $\delta$-statistical parity and $g$ is $\epsilon$-noncomparatively fair with respect to $f$, we have
\begin{equation}
\begin{array}{lcl}
||p_{x, y}(g, a) - p_{x, y}(g, a')|| < 2M\epsilon + \delta
\end{array}
\end{equation}

However, our goal is to identify auditors who satisfy $\delta'$-statistical parity. In other words, we need
\begin{equation}
\begin{array}{lcl}
\delta' > 2M\epsilon + \delta.
\end{array}
\end{equation}

Upon rearranging the terms, we obtain the bound on $\epsilon$ as stated in the corollary statement.

\end{proof}
Thus, any given auditor can be evaluated in terms of traditional fairness notions (e.g. individual fairness, statistical parity) using our proposed non-comparative fairness notion. Once the biases within a benchmark system is quantified, we can identify fair auditors via screening them using a threshold-based rule on their biases from a non-comparative fairness perspective. Note that the threshold used in this classifier can be computed based on the designer's bias tolerance. Having identified fair auditors along with quantifying their respective biases, we can leverage them as benchmark entities to evaluate an unknown recommender system via computing the system's biases using the auditor's evaluation.

\section{Conclusion and Future Work}
We introduced a non-comparative fairness notion which complements the existing comparative fairness notions proposed in the algorithmic fairness literature. We showed that any system can be deemed fair from the perspective of comparative fairness (e.g. individual fairness and statistical parity) if it is non-comparatively fair with respect to an auditor who has been deemed fair with respect to the same fairness notions. We also proved that the converse holds true in the context of individual fairness. We discussed how the proposed non-comparative fairness notion can be used to identify fair auditors who are hired to evaluate latent biases in decision-support systems. In our future work, we will develop novel algorithms to identify fair auditors from the perspective of multiple (potentially intransitive) attributes, and also validate our theoretical findings using real data.

\bibliographystyle{ACM-Reference-Format}
\bibliography{sample-bibliography}


\begin{thebibliography}{33}


\ifx \showCODEN    \undefined \def \showCODEN     #1{\unskip}     \fi
\ifx \showDOI      \undefined \def \showDOI       #1{#1}\fi
\ifx \showISBNx    \undefined \def \showISBNx     #1{\unskip}     \fi
\ifx \showISBNxiii \undefined \def \showISBNxiii  #1{\unskip}     \fi
\ifx \showISSN     \undefined \def \showISSN      #1{\unskip}     \fi
\ifx \showLCCN     \undefined \def \showLCCN      #1{\unskip}     \fi
\ifx \shownote     \undefined \def \shownote      #1{#1}          \fi
\ifx \showarticletitle \undefined \def \showarticletitle #1{#1}   \fi
\ifx \showURL      \undefined \def \showURL       {\relax}        \fi
\providecommand\bibfield[2]{#2}
\providecommand\bibinfo[2]{#2}
\providecommand\natexlab[1]{#1}
\providecommand\showeprint[2][]{arXiv:#2}

\bibitem[\protect\citeauthoryear{Angwin, Larson, Mattu, and Kirchner}{Angwin
  et~al\mbox{.}}{2016}]%
        {angwin2016}
\bibfield{author}{\bibinfo{person}{J. Angwin}, \bibinfo{person}{J. Larson},
  \bibinfo{person}{S. Mattu}, {and} \bibinfo{person}{L. Kirchner}.}
  \bibinfo{year}{2016}\natexlab{}.
\newblock \showarticletitle{{Machine Bias}}.
\newblock \bibinfo{journal}{\emph{ProPublica}} (\bibinfo{date}{May 23}
  \bibinfo{year}{2016}).
\newblock


\bibitem[\protect\citeauthoryear{Barry-Jester, Ben, and Dana}{Barry-Jester
  et~al\mbox{.}}{2015}]%
        {baryjester2015}
\bibfield{author}{\bibinfo{person}{Anna~Maria Barry-Jester},
  \bibinfo{person}{C. Ben}, {and} \bibinfo{person}{G. Dana}.}
  \bibinfo{year}{2015}\natexlab{}.
\newblock \showarticletitle{{The New Science of Sentencing}}.
\newblock \bibinfo{journal}{\emph{The Marshall Project}} (\bibinfo{date}{August
  08} \bibinfo{year}{2015}).
\newblock


\bibitem[\protect\citeauthoryear{Binns}{Binns}{2019}]%
        {binns2019human}
\bibfield{author}{\bibinfo{person}{Reuben Binns}.}
  \bibinfo{year}{2019}\natexlab{}.
\newblock \showarticletitle{Human Judgement in Algorithmic Loops; Individual
  Justice and Automated Decision-Making}.
\newblock \bibinfo{journal}{\emph{Individual Justice and Automated
  Decision-Making (September 11, 2019)}} (\bibinfo{year}{2019}).
\newblock


\bibitem[\protect\citeauthoryear{Chen, Kallus, Mao, Svacha, and Udell}{Chen
  et~al\mbox{.}}{2019}]%
        {chen2019fairness}
\bibfield{author}{\bibinfo{person}{Jiahao Chen}, \bibinfo{person}{Nathan
  Kallus}, \bibinfo{person}{Xiaojie Mao}, \bibinfo{person}{Geoffry Svacha},
  {and} \bibinfo{person}{Madeleine Udell}.} \bibinfo{year}{2019}\natexlab{}.
\newblock \showarticletitle{{Fairness Under Unawareness: Assessing Disparity
  When Protected Class Is Unobserved}}. In
  \bibinfo{booktitle}{\emph{{Proceedings of the Conference on Fairness,
  Accountability, and Transparency}}}. ACM, \bibinfo{pages}{339--348}.
\newblock


\bibitem[\protect\citeauthoryear{Chouldechova and Roth}{Chouldechova and
  Roth}{2018}]%
        {chouldechova2018frontiers}
\bibfield{author}{\bibinfo{person}{Alexandra Chouldechova} {and}
  \bibinfo{person}{Aaron Roth}.} \bibinfo{year}{2018}\natexlab{}.
\newblock \showarticletitle{The frontiers of fairness in machine learning}.
\newblock \bibinfo{journal}{\emph{arXiv preprint arXiv:1810.08810}}
  (\bibinfo{year}{2018}).
\newblock


\bibitem[\protect\citeauthoryear{Dwork, Hardt, Pitassi, Reingold, and
  Zemel}{Dwork et~al\mbox{.}}{2012}]%
        {dwork2012fairness}
\bibfield{author}{\bibinfo{person}{C. Dwork}, \bibinfo{person}{M. Hardt},
  \bibinfo{person}{T. Pitassi}, \bibinfo{person}{O. Reingold}, {and}
  \bibinfo{person}{R. Zemel}.} \bibinfo{year}{2012}\natexlab{}.
\newblock \showarticletitle{{Fairness Through Awareness}}. In
  \bibinfo{booktitle}{\emph{Proceedings of the 3rd innovations in theoretical
  computer science conference}}. ACM, \bibinfo{pages}{214--226}.
\newblock


\bibitem[\protect\citeauthoryear{Feinberg}{Feinberg}{1974}]%
        {feinberg1974noncomparative}
\bibfield{author}{\bibinfo{person}{Joel Feinberg}.}
  \bibinfo{year}{1974}\natexlab{}.
\newblock \showarticletitle{Noncomparative justice}.
\newblock \bibinfo{journal}{\emph{The philosophical review}}
  \bibinfo{volume}{83}, \bibinfo{number}{3} (\bibinfo{year}{1974}),
  \bibinfo{pages}{297--338}.
\newblock


\bibitem[\protect\citeauthoryear{Fisman and Luca}{Fisman and Luca}{2016}]%
        {Harvard2016}
\bibfield{author}{\bibinfo{person}{Ray Fisman} {and} \bibinfo{person}{Michael
  Luca}.} \bibinfo{year}{2016}\natexlab{}.
\newblock \showarticletitle{{Fixing Discrimination in Online Marketplaces}}. In
  \bibinfo{booktitle}{\emph{Harvard Business Review}}.
\newblock


\bibitem[\protect\citeauthoryear{Foundation}{Foundation}{1999}]%
        {colorofcrime}
\bibfield{author}{\bibinfo{person}{New~Century Foundation}.}
  \bibinfo{year}{1999}\natexlab{}.
\newblock \showarticletitle{The Color of Crime - Race, Crime and Violence in
  America}.
\newblock  (\bibinfo{year}{1999}).
\newblock
\urldef\tempurl%
\url{https://2kpcwh2r7phz1nq4jj237m22-wpengine.netdna-ssl.com/wp-content/uploads/2011/12/1999-Color-of-Crime-Report.pdf}
\showURL{%
\tempurl}


\bibitem[\protect\citeauthoryear{Gillen, Jung, Kearns, and Roth}{Gillen
  et~al\mbox{.}}{2018}]%
        {gillen2018online}
\bibfield{author}{\bibinfo{person}{Stephen Gillen},
  \bibinfo{person}{Christopher Jung}, \bibinfo{person}{Michael Kearns}, {and}
  \bibinfo{person}{Aaron Roth}.} \bibinfo{year}{2018}\natexlab{}.
\newblock \showarticletitle{Online learning with an unknown fairness metric}.
  In \bibinfo{booktitle}{\emph{Advances in Neural Information Processing
  Systems}}. \bibinfo{pages}{2600--2609}.
\newblock


\bibitem[\protect\citeauthoryear{Greenberg}{Greenberg}{1986}]%
        {greenberg1986determinants}
\bibfield{author}{\bibinfo{person}{Jerald Greenberg}.}
  \bibinfo{year}{1986}\natexlab{}.
\newblock \showarticletitle{Determinants of perceived fairness of performance
  evaluations.}
\newblock \bibinfo{journal}{\emph{Journal of applied psychology}}
  \bibinfo{volume}{71}, \bibinfo{number}{2} (\bibinfo{year}{1986}),
  \bibinfo{pages}{340}.
\newblock


\bibitem[\protect\citeauthoryear{Grgic-Hlaca, Redmiles, Gummadi, and
  Weller}{Grgic-Hlaca et~al\mbox{.}}{2018}]%
        {grgic2018human}
\bibfield{author}{\bibinfo{person}{Nina Grgic-Hlaca}, \bibinfo{person}{Elissa~M
  Redmiles}, \bibinfo{person}{Krishna~P Gummadi}, {and} \bibinfo{person}{Adrian
  Weller}.} \bibinfo{year}{2018}\natexlab{}.
\newblock \showarticletitle{Human perceptions of fairness in algorithmic
  decision making: A case study of criminal risk prediction}. In
  \bibinfo{booktitle}{\emph{Proceedings of the 2018 World Wide Web
  Conference}}. \bibinfo{pages}{903--912}.
\newblock


\bibitem[\protect\citeauthoryear{Hardt, Price, and Srebro}{Hardt
  et~al\mbox{.}}{2016}]%
        {MoritzOpportunities}
\bibfield{author}{\bibinfo{person}{M. Hardt}, \bibinfo{person}{E. Price}, {and}
  \bibinfo{person}{N. Srebro}.} \bibinfo{year}{2016}\natexlab{}.
\newblock \showarticletitle{{Equality of Opportunity in Supervised Learning}}.
\newblock In \bibinfo{booktitle}{\emph{Advances in Neural Information
  Processing Systems 29}}, \bibfield{editor}{\bibinfo{person}{D.~D. Lee},
  \bibinfo{person}{M.~Sugiyama}, \bibinfo{person}{U.~V. Luxburg},
  \bibinfo{person}{I.~Guyon}, {and} \bibinfo{person}{R.~Garnett}} (Eds.).
  \bibinfo{publisher}{Curran Associates, Inc.}, \bibinfo{pages}{3315--3323}.
\newblock


\bibitem[\protect\citeauthoryear{{Jon Fingas}}{{Jon Fingas}}{2019}]%
        {applecard}
\bibfield{author}{\bibinfo{person}{{Jon Fingas}}.}
  \bibinfo{year}{2019}\natexlab{}.
\newblock \bibinfo{title}{New York Investigates Claims of Sexism in Apple Card
  Credit Limits}.
\newblock
  \bibinfo{howpublished}{\url{https://www.engadget.com/2019-11-09-new-york-investigates-apple-card-credit-limit-sexism.html}}.
\newblock


\bibitem[\protect\citeauthoryear{Jung, Kearns, Neel, Roth, Stapleton, and
  Wu}{Jung et~al\mbox{.}}{2019}]%
        {jung2019eliciting}
\bibfield{author}{\bibinfo{person}{Christopher Jung}, \bibinfo{person}{Michael
  Kearns}, \bibinfo{person}{Seth Neel}, \bibinfo{person}{Aaron Roth},
  \bibinfo{person}{Logan Stapleton}, {and} \bibinfo{person}{Zhiwei~Steven Wu}.}
  \bibinfo{year}{2019}\natexlab{}.
\newblock \showarticletitle{Eliciting and enforcing subjective individual
  fairness}.
\newblock \bibinfo{journal}{\emph{arXiv preprint arXiv:1905.10660}}
  (\bibinfo{year}{2019}).
\newblock


\bibitem[\protect\citeauthoryear{Kearns, Neel, Roth, and Wu}{Kearns
  et~al\mbox{.}}{2017}]%
        {kearns2017preventing}
\bibfield{author}{\bibinfo{person}{Michael Kearns}, \bibinfo{person}{Seth
  Neel}, \bibinfo{person}{Aaron Roth}, {and} \bibinfo{person}{Zhiwei~Steven
  Wu}.} \bibinfo{year}{2017}\natexlab{}.
\newblock \showarticletitle{Preventing fairness gerrymandering: Auditing and
  learning for subgroup fairness}.
\newblock \bibinfo{journal}{\emph{arXiv preprint arXiv:1711.05144}}
  (\bibinfo{year}{2017}).
\newblock


\bibitem[\protect\citeauthoryear{Kim, Reingold, and Rothblum}{Kim
  et~al\mbox{.}}{2018}]%
        {kim2018fairness}
\bibfield{author}{\bibinfo{person}{Michael Kim}, \bibinfo{person}{Omer
  Reingold}, {and} \bibinfo{person}{Guy Rothblum}.}
  \bibinfo{year}{2018}\natexlab{}.
\newblock \showarticletitle{Fairness through computationally-bounded
  awareness}. In \bibinfo{booktitle}{\emph{Advances in Neural Information
  Processing Systems}}. \bibinfo{pages}{4842--4852}.
\newblock


\bibitem[\protect\citeauthoryear{Koren}{Koren}{2016}]%
        {koren2016}
\bibfield{author}{\bibinfo{person}{James~Rufus Koren}.}
  \bibinfo{year}{2016}\natexlab{}.
\newblock \showarticletitle{{What Does That Web Search Say About Your Credit?}}
\newblock \bibinfo{journal}{\emph{Los Angeles Times}} (\bibinfo{date}{July 17}
  \bibinfo{year}{2016}).
\newblock


\bibitem[\protect\citeauthoryear{Kusner, Loftus, Russell, and Silva}{Kusner
  et~al\mbox{.}}{2017}]%
        {kusner2017}
\bibfield{author}{\bibinfo{person}{Matt~J Kusner}, \bibinfo{person}{J. Loftus},
  \bibinfo{person}{C. Russell}, {and} \bibinfo{person}{R. Silva}.}
  \bibinfo{year}{2017}\natexlab{}.
\newblock \showarticletitle{{Counterfactual Fairness}}. In
  \bibinfo{booktitle}{\emph{Advances in Neural Information Processing
  Systems}}. \bibinfo{pages}{4066--4076}.
\newblock


\bibitem[\protect\citeauthoryear{Lepri, Oliver, Letouz{\'e}, Pentland, and
  Vinck}{Lepri et~al\mbox{.}}{2018}]%
        {Lepri2018}
\bibfield{author}{\bibinfo{person}{B. Lepri}, \bibinfo{person}{N. Oliver},
  \bibinfo{person}{E. Letouz{\'e}}, \bibinfo{person}{A. Pentland}, {and}
  \bibinfo{person}{P. Vinck}.} \bibinfo{year}{2018}\natexlab{}.
\newblock \showarticletitle{{Fair, Transparent, and Accountable Algorithmic
  Decision-making Processes}}.
\newblock \bibinfo{journal}{\emph{Philosophy {\&} Technology}}
  \bibinfo{volume}{31}, \bibinfo{number}{4} (\bibinfo{date}{01 Dec}
  \bibinfo{year}{2018}), \bibinfo{pages}{611--627}.
\newblock


\bibitem[\protect\citeauthoryear{Levine and Pannier}{Levine and
  Pannier}{2005}]%
        {levine2005comparative}
\bibfield{author}{\bibinfo{person}{Raleigh~Hannah Levine} {and}
  \bibinfo{person}{Russell Pannier}.} \bibinfo{year}{2005}\natexlab{}.
\newblock \showarticletitle{Comparative and noncomparative justice: some
  guidelines for constitutional adjudication}.
\newblock \bibinfo{journal}{\emph{Wm. \& Mary Bill Rts. J.}}
  \bibinfo{volume}{14} (\bibinfo{year}{2005}), \bibinfo{pages}{141}.
\newblock


\bibitem[\protect\citeauthoryear{Montague}{Montague}{1980}]%
        {montague1980comparative}
\bibfield{author}{\bibinfo{person}{Phillip Montague}.}
  \bibinfo{year}{1980}\natexlab{}.
\newblock \showarticletitle{Comparative and non-comparative justice}.
\newblock \bibinfo{journal}{\emph{The Philosophical Quarterly (1950-)}}
  \bibinfo{volume}{30}, \bibinfo{number}{119} (\bibinfo{year}{1980}),
  \bibinfo{pages}{131--140}.
\newblock


\bibitem[\protect\citeauthoryear{Raji, Smart, White, Mitchell, Gebru,
  Hutchinson, Smith-Loud, Theron, and Barnes}{Raji et~al\mbox{.}}{2020}]%
        {raji2020closing}
\bibfield{author}{\bibinfo{person}{Inioluwa~Deborah Raji},
  \bibinfo{person}{Andrew Smart}, \bibinfo{person}{Rebecca~N White},
  \bibinfo{person}{Margaret Mitchell}, \bibinfo{person}{Timnit Gebru},
  \bibinfo{person}{Ben Hutchinson}, \bibinfo{person}{Jamila Smith-Loud},
  \bibinfo{person}{Daniel Theron}, {and} \bibinfo{person}{Parker Barnes}.}
  \bibinfo{year}{2020}\natexlab{}.
\newblock \showarticletitle{Closing the AI accountability gap: defining an
  end-to-end framework for internal algorithmic auditing}. In
  \bibinfo{booktitle}{\emph{Proceedings of the 2020 Conference on Fairness,
  Accountability, and Transparency}}. \bibinfo{pages}{33--44}.
\newblock


\bibitem[\protect\citeauthoryear{Ritov, Sun, and Zhao}{Ritov
  et~al\mbox{.}}{2017}]%
        {ritov2017conditional}
\bibfield{author}{\bibinfo{person}{Ya'acov Ritov}, \bibinfo{person}{Yuekai
  Sun}, {and} \bibinfo{person}{Ruofei Zhao}.} \bibinfo{year}{2017}\natexlab{}.
\newblock \showarticletitle{On conditional parity as a notion of
  non-discrimination in machine learning}.
\newblock \bibinfo{journal}{\emph{arXiv preprint arXiv:1706.08519}}
  (\bibinfo{year}{2017}).
\newblock


\bibitem[\protect\citeauthoryear{Saxena, Huang, DeFilippis, Radanovic, Parkes,
  and Liu}{Saxena et~al\mbox{.}}{2019}]%
        {saxena2019fairness}
\bibfield{author}{\bibinfo{person}{Nripsuta~Ani Saxena}, \bibinfo{person}{Karen
  Huang}, \bibinfo{person}{Evan DeFilippis}, \bibinfo{person}{Goran Radanovic},
  \bibinfo{person}{David~C Parkes}, {and} \bibinfo{person}{Yang Liu}.}
  \bibinfo{year}{2019}\natexlab{}.
\newblock \showarticletitle{How do fairness definitions fare? Examining public
  attitudes towards algorithmic definitions of fairness}. In
  \bibinfo{booktitle}{\emph{Proceedings of the 2019 AAAI/ACM Conference on AI,
  Ethics, and Society}}. \bibinfo{pages}{99--106}.
\newblock


\bibitem[\protect\citeauthoryear{Smith and Linden}{Smith and Linden}{2017}]%
        {smith2017two}
\bibfield{author}{\bibinfo{person}{Brent Smith} {and} \bibinfo{person}{Greg
  Linden}.} \bibinfo{year}{2017}\natexlab{}.
\newblock \showarticletitle{{Two Decades of Recommender Systems at
  Amazon.com}}.
\newblock \bibinfo{journal}{\emph{IEEE Internet Computing}}
  \bibinfo{volume}{21}, \bibinfo{number}{3} (\bibinfo{year}{2017}),
  \bibinfo{pages}{12--18}.
\newblock


\bibitem[\protect\citeauthoryear{Speicher, Heidari, Grgic-Hlaca, Gummadi,
  Singla, Weller, and Zafar}{Speicher et~al\mbox{.}}{2018}]%
        {speicher2018unified}
\bibfield{author}{\bibinfo{person}{Till Speicher}, \bibinfo{person}{Hoda
  Heidari}, \bibinfo{person}{Nina Grgic-Hlaca}, \bibinfo{person}{Krishna~P
  Gummadi}, \bibinfo{person}{Adish Singla}, \bibinfo{person}{Adrian Weller},
  {and} \bibinfo{person}{Muhammad~Bilal Zafar}.}
  \bibinfo{year}{2018}\natexlab{}.
\newblock \showarticletitle{A unified approach to quantifying algorithmic
  unfairness: Measuring individual \&group unfairness via inequality indices}.
  In \bibinfo{booktitle}{\emph{Proceedings of the 24th ACM SIGKDD International
  Conference on Knowledge Discovery \& Data Mining}}.
  \bibinfo{pages}{2239--2248}.
\newblock


\bibitem[\protect\citeauthoryear{Srivastava, Heidari, and Krause}{Srivastava
  et~al\mbox{.}}{2019}]%
        {srivastava2019mathematical}
\bibfield{author}{\bibinfo{person}{Megha Srivastava}, \bibinfo{person}{Hoda
  Heidari}, {and} \bibinfo{person}{Andreas Krause}.}
  \bibinfo{year}{2019}\natexlab{}.
\newblock \showarticletitle{Mathematical notions vs. human perception of
  fairness: A descriptive approach to fairness for machine learning}. In
  \bibinfo{booktitle}{\emph{Proceedings of the 25th ACM SIGKDD International
  Conference on Knowledge Discovery \& Data Mining}}.
  \bibinfo{pages}{2459--2468}.
\newblock


\bibitem[\protect\citeauthoryear{Waxman}{Waxman}{2018}]%
        {waxmen2018}
\bibfield{author}{\bibinfo{person}{Andrew Waxman}.}
  \bibinfo{year}{2018}\natexlab{}.
\newblock \showarticletitle{{BankThink AI Can Help Banks Make Better Decisions,
  But it Doesn't Remove Bias}}.
\newblock \bibinfo{journal}{\emph{American Banker}} (\bibinfo{date}{June 05}
  \bibinfo{year}{2018}).
\newblock


\bibitem[\protect\citeauthoryear{{Wikipedia contributors}}{{Wikipedia
  contributors}}{2020}]%
        {reservationindia}
\bibfield{author}{\bibinfo{person}{{Wikipedia contributors}}.}
  \bibinfo{year}{2020}\natexlab{}.
\newblock \bibinfo{title}{Reservation in India --- {Wikipedia}{,} The Free
  Encyclopedia}.
\newblock
  \bibinfo{howpublished}{\url{https://en.wikipedia.org/w/index.php?title=Reservation_in_India&oldid=969093785}}.
\newblock
\newblock
\shownote{[Online; accessed 29-July-2020].}


\bibitem[\protect\citeauthoryear{Zafar, Valera, Gomez~Rodriguez, and
  Gummadi}{Zafar et~al\mbox{.}}{2017a}]%
        {zafar2017fairness}
\bibfield{author}{\bibinfo{person}{Muhammad~Bilal Zafar},
  \bibinfo{person}{Isabel Valera}, \bibinfo{person}{Manuel Gomez~Rodriguez},
  {and} \bibinfo{person}{Krishna~P Gummadi}.} \bibinfo{year}{2017}\natexlab{a}.
\newblock \showarticletitle{Fairness beyond disparate treatment \& disparate
  impact: Learning classification without disparate mistreatment}. In
  \bibinfo{booktitle}{\emph{Proceedings of the 26th International Conference on
  World Wide Web}}. International World Wide Web Conferences Steering
  Committee, \bibinfo{pages}{1171--1180}.
\newblock


\bibitem[\protect\citeauthoryear{Zafar, Valera, Rodriguez, Gummadi, and
  Weller}{Zafar et~al\mbox{.}}{2017b}]%
        {zafar2017}
\bibfield{author}{\bibinfo{person}{Muhammad~Bilal Zafar}, \bibinfo{person}{I.
  Valera}, \bibinfo{person}{M. Rodriguez}, \bibinfo{person}{K. Gummadi}, {and}
  \bibinfo{person}{A. Weller}.} \bibinfo{year}{2017}\natexlab{b}.
\newblock \showarticletitle{{From Parity to Preference-based Notions of
  Fairness in Classification}}.
\newblock In \bibinfo{booktitle}{\emph{Advances in Neural Information
  Processing Systems 30}}, \bibfield{editor}{\bibinfo{person}{I.~Guyon},
  \bibinfo{person}{U.~V. Luxburg}, \bibinfo{person}{S.~Bengio},
  \bibinfo{person}{H.~Wallach}, \bibinfo{person}{R.~Fergus},
  \bibinfo{person}{S.~Vishwanathan}, {and} \bibinfo{person}{R.~Garnett}}
  (Eds.). \bibinfo{publisher}{Curran Associates, Inc.},
  \bibinfo{pages}{229--239}.
\newblock


\bibitem[\protect\citeauthoryear{Zhang and Neill}{Zhang and Neill}{2016}]%
        {zhang2016identifying}
\bibfield{author}{\bibinfo{person}{Zhe Zhang} {and} \bibinfo{person}{Daniel~B.
  Neill}.} \bibinfo{year}{2016}\natexlab{}.
\newblock \bibinfo{title}{Identifying Significant Predictive Bias in
  Classifiers}.
\newblock
\newblock
\showeprint[arxiv]{stat.ML/1611.08292}


\end{thebibliography}

\end{document}